\documentclass[journal]{IEEEtran}
\usepackage{amsmath,amssymb,amsbsy,amsfonts,amsthm,latexsym,
               amsopn,amstext,amsxtra,euscript,amscd,color,multirow}

\usepackage[margin=0.8in]{geometry}
\usepackage{geometry}
\usepackage{enumerate}
\usepackage{color}
\usepackage{longtable}
\usepackage{caption2}

\usepackage{float}

\floatstyle{plaintop}
\restylefloat{table}

\allowdisplaybreaks[1]

\newcommand*\colvec[3][]{
    \begin{pmatrix}\ifx\relax#1\relax\else#1\\\fi#2\\#3\end{pmatrix}
}

\newtheorem{theorem}{Theorem}
\newtheorem{corollary}[theorem]{Corollary}
\newtheorem{lemma}[theorem]{Lemma}
\newtheorem{proposition}[theorem]{Proposition}
\newtheorem{definition}[theorem]{Definition}

\newtheorem{result}[theorem]{Result}

\numberwithin{theorem}{section}
\numberwithin{equation}{section}
\numberwithin{table}{section}

\newtheorem{remark}[theorem]{Remark}

\newcommand{\cA}{\mathcal{A}}
\newcommand{\cB}{\mathcal{B}}

\newcommand{\al}{\alpha}

\newcommand{\be}{\beta}
\newcommand{\de}{\delta}
\newcommand{\De}{\Delta}
\newcommand{\ga}{\gamma}

\newcommand{\ze}{\zeta}

\newcommand{\sig}{\sigma}

\newcommand{\sm}{\setminus}
\newcommand{\es}{\emptyset}

\newcommand{\rank}{\text{rank}}

\newcommand{\lf}{\lfloor}
\newcommand{\rf}{\rfloor}

\newcommand{\F}{\mathbb{F}}

\newcommand{\Fq}{\mathbb{F}_q}

\newcommand{\Gal}{\text{Gal}}

\newcommand{\VB}{\mathcal{VB}}

\newcommand{\Ftn}{\mathbb{F}_{2^n}}
\newcommand{\Fts}{\mathbb{F}_{2^s}}

\long\def\symbolfootnote[#1]#2{\begingroup%
\def\thefootnote{\fnsymbol{footnote}}\footnote[#1]{#2}\endgroup}

\usepackage[colorlinks=true,citecolor=blue,linkcolor=blue,urlcolor=blue,bookmarks,bookmarksopen,bookmarksdepth=2,backref=page]{hyperref}

\renewcommand*{\backref}[1]{}
\renewcommand*{\backrefalt}[4]{%
	\ifcase #1 (Not cited.)%
	\or        (Cited on page~#2.)%
	\else      (Cited on pages~#2.)%
	\fi}

\let\OLDthebibliography\thebibliography
\renewcommand\thebibliography[1]{
	\OLDthebibliography{#1}
	\setlength{\parskip}{0pt}
	\setlength{\itemsep}{0pt plus 0.3ex}
}

\begin{document}

\title{Vanishing Flats: A Combinatorial Viewpoint on the Planarity of Functions and Their Application}
\author{Shuxing Li, Wilfried Meidl, Alexandr Polujan, Alexander Pott, Constanza Riera, and Pantelimon St\u anic\u a
\thanks{S. Li is with Institute of Algebra and Geometry, Faculty of Mathematics, Otto von Guericke University Magdeburg, Universit\"atsplatz 2, 39106 Magdeburg, Germany (email: shuxing\_li@sfu.ca).}
\thanks{W. Meidl is with Johann Radon Institute for Computational and Applied Mathematics, Austrian Academy of Sciences, Altenbergerstrasse 69, 4040 Linz, Austria (email: meidlwilfried@gmail.com).}
\thanks{A. Polujan is with Institute of Algebra and Geometry, Faculty of Mathematics, Otto von Guericke University Magdeburg, Universit\"atsplatz 2, 39106 Magdeburg, Germany (email: alexandr.polujan@ovgu.de).}
\thanks{A. Pott is with Institute of Algebra and Geometry, Faculty of Mathematics, Otto von Guericke University Magdeburg, Universit\"atsplatz 2, 39106 Magdeburg, Germany (email: alexander.pott@ovgu.de).}
\thanks{C. Riera is with Department of Computer Science, Electrical Engineering and Mathematical Sciences, Western Norway University of Applied Sciences, 5020 Bergen, Norway (email: csr@hvl.no).}
\thanks{P. St\u anic\u a is with Department of Applied Mathematics, Naval Postgraduate School, Monterey, CA 93943--5216 (email: pstanica@nps.edu).}
}

\maketitle

\begin{abstract}
For a function $f$ from $\F_2^n$ to $\F_2^n$, the planarity of $f$ is usually measured by its differential uniformity and differential spectrum. In this paper, we propose the concept of vanishing flats, which supplies a combinatorial viewpoint on the planarity. First, the number of vanishing flats of $f$ can be regarded as a measure of the distance between $f$ and the set of almost perfect nonlinear functions. In some cases, the number of vanishing flats serves as an ``intermediate'' concept between differential uniformity and differential spectrum, which contains more information than differential uniformity, however less than the differential spectrum. Secondly, the set of vanishing flats forms a combinatorial configuration called partial quadruple system, since it conveys a detailed structural information about $f$. We initiate this study by considering the number of vanishing flats and the partial quadruple systems associated with monomials and Dembowski-Ostrom polynomials. In addition, we present an application of vanishing flats to the partition of a vector space into disjoint equidimensional affine spaces. We conclude the paper with several further questions and challenges.
\end{abstract}


\section{Introduction}

In this paper, we propose a combinatorial viewpoint to assess the planarity of functions from $\F_2^n$ to $\F_2^n$.  A function $f: \F_2^n \rightarrow \F_2^n$ achieves the highest planarity if and only if for each nonzero $a \in \F_2^n$, the function $f(x+a)-f(x)$ is as balanced as possible. A function from $\F_2^n$ to $\F_2^n$ with the highest planarity is called almost perfect nonlinear and will be formally defined later.  Keeping in mind that planarity is defined for functions over vector spaces, we always adopt the finite field model $\Ftn$ to represent the vector space $\F_2^n$, as each mapping defined over $\Ftn$ can be described by a polynomial. In fact, the univariate polynomial representation of functions over $\F_{2^n}$ may better serve our purpose and will be used throughout.

We first recall some basic concepts and notation. For a polynomial $f$ over $\Ftn$ and $a \in \Ftn^*$, the \emph{first-order derivative of $f$ at $a$} is the polynomial $f(x+a)+f(x)$. Given a polynomial $f$ and field elements $a \in \Ftn^*$, $b \in \Ftn$, we define
$$
\de_f(a,b)=|\{ x \in \Ftn \mid f(x+a)+f(x)=b \}|.
$$
Clearly, $\de_f(a,b)$ must be even (since, if $x_0$ is a solution of the above equation, so is $x_0+a$). The \emph{differential uniformity} of $f$ is defined as
$$
\de_f=\max_{a \in \Ftn^*, b \in \Ftn} \de_f(a,b)
$$
and the \emph{differential spectrum} of $f$ is the sequence $(\ell_{f,0},\ell_{f,2},\ldots,\ell_{f,\de_f})$, where $\ell_{f,2i}$ is the frequency of $2i$ in the multiset $[\de_f(a,b) \mid a \in \Ftn^*, b \in \Ftn]$. For the sake of convenience, we use $\ell_{2i}$ to represent $\ell_{f,2i}$, whenever the function $f$ is clear from the context. The polynomial $f$ is called {\em almost perfect nonlinear} (APN), if its differential uniformity is $\de_f=2$, which is the smallest possible. Equivalently, $f$ is APN if and only if the first-order derivative $f(x+a)+f(x)$ is a $2$-to-$1$ mapping for each $a \in \Ftn^*$. For $a \in \Ftn^*$, we define the \emph{differential uniformity of $f$ along the direction $a$} as
\[
\de_f(a)=\max_{b \in \Ftn} \de_f(a,b).
\]
Moreover, for $a \in \Ftn^*$, we define
\[
E_f(a)=\left\{ f(x+a)+f(x) \mid x \in \Ftn \right\}
\]
to be the image set of the first-order derivative of $f$ along the direction $a$. Note that $\de_f(a) \ge 2$ for each $a \in \Ftn^*$ and $f$ is APN if and only if $\de_f(a)=2$ for each $a \in \Ftn^*$. The recent paper~\cite{BKKRS} calls a function $f$ \emph{partially $a$-APN}, if $\de_f(a)=2$ at some $a \in \Ftn^*$. Surely, $f$ is partially $a$-APN if and only if $|E_f(a)|=2^{n-1}$.


There has been intensive research about APN functions, and  we refer to~\cite{BN15} and \cite{P}, for comprehensive surveys and extensive references on the subject. The fact that most known APN polynomials are quadratic and there is only one known class of APN permutations over $\Ftn$ with $n=6$, further motivates our research into the ``second best'' polynomials, which are not APN, but are closest to APN functions in some sense (this was one of the reasons for the introduction of the previously mentioned partially APN functions). One natural research venue  is  finding polynomials with differential uniformity $4$, which is the second smallest possible. Along this direction, a considerable amount of research dealt with polynomials having differential uniformity $4$, also satisfying some other  properties (see~\cite{BCC,BL,BTT,N,QTTL},  for example). The main purpose of this paper is to propose yet another research venue, namely the concept of vanishing flats, which offers a combinatorial viewpoint to understand the planarity of functions. The number of vanishing flats serves as a measure of the planarity of functions, or equivalently, measure the distance from a function to the set of APN functions. The set of vanishing flats forms a configuration called partial quadruple system, which is closely related to the classical Steiner quadruple system. The partial quadruple system contains detailed structural information of $f$. For instance, one can easily determine the directions along which a function $f$ is partially APN, by reading through the partial quadruple system of $f$.

The rest of the paper is organized as follows. In Section~\ref{sec2}, we introduce the vanishing flats and the partial quadruple system. Both the number of vanishing flats and the partial quadruple system are invariants under the CCZ-equivalence~\cite{CCZ}, as we shall show. Some relations among the differential uniformity, the number of vanishing flats and differential spectrum are discussed. With respect to the partial quadruple systems, some combinatorial implications of differential uniformity and differential spectrum are revealed. In Section~\ref{sec3}, we consider the partial quadruple system and the number of vanishing flats associated with monomials and Dembowski-Ostrom polynomials. We determine the partial quadruple systems of the inverse and the Gold functions. For general Dembowski-Ostrom polynomials, a formula for the number of vanishing flats is presented. In addition, the number of vanishing flats of a few more classes of monomials follows directly from their differential spectra. Section~\ref{sec4} presents an application of vanishing flats to a problem of partitioning vector spaces over finite fields into a collection of disjoint equidimensional affine subspaces, in the spirit of Baum and Neuwirth~\cite{BN75}. By extending the notion of vanishing flats, we propose a new approach to generate this type of partition. In Section~\ref{sec5}, we give some concluding remarks and raise some future research problems.

\section{Vanishing flats and partial quadruple systems associated with Boolean functions}\label{sec2}

In this section, we propose the concept of vanishing flats and a configuration named partial quadruple system associated with non-APN functions.
For a polynomial $f: \Ftn \rightarrow \Ftn$, we define the set of \emph{critical directions} of $f$ as
\begin{align*}
D_f&=\{ a \in \Ftn^* \mid \mbox{$f(x+a)+f(x)$ is not $2$-to-$1$} \}\\
   &=\{a \in \Ftn^* \mid \de_f(a) \ge 4\}.
\end{align*}
Clearly, $f$ is APN if and only if $D_f=\es$, namely, there exists no critical direction of $f$, or, the first-order derivative is as much balanced as possible along each direction. Indeed, in order to verify the APN property, elegant results in \cite{BD,CK} say that checking about half of all the directions would suffice. More precisely, if $\Ftn^* \sm D_f$ contains an $(n-1)$-dimensional vector space over $\F_2$ minus the zero vector, then this much weaker condition ensures that $f$ is APN \cite[Theorem 2]{CK} (this result was extended in~\cite[Theorem 6.2]{CP} to show that if $\Ftn^* \sm D_f$ contains a $t$-dimensional vector space over $\F_2$ minus the zero vector then the differential uniformity is $\leq 2^{n-t}$). Surely, for a non-APN function $f$, the set $D_f$ indicates the critical directions where the $2$-to-$1$ property fails. For instance, for $a \in \Ftn^*$, a function $f$ is partially $a$-APN if and only if $a \not\in D_f$. Consequently, we consider the set $D_f$ to be a crucial ingredient conveying much information related to the planarity of $f$. In order to understand $D_f$, we trace back to its source and arrive at a well-known alternative definition of APN functions.


For $n \ge 2$, define the set of all $2$-dimensional flats in $\Ftn$ as
\begin{align*}
\cB_n=\{ \{x_1,x_2,x_3, & x_4\}  \mid \mbox{$x_1+x_2+x_3+x_4=0$ } \\
                             & \mbox{and $x_1,x_2,x_3,x_4 \in\Ftn$ are distinct} \}.
\end{align*}
Conventionally, each subset in $\cB_n$ is called a \emph{block}. The classical \emph{Steiner quadruple system} is a pair $(\Ftn,\cB_n)$, so that each $3$-subset of $\Ftn$ is contained in exactly one block of $\cB_n$. A function $f: \Ftn \rightarrow \Ftn$ is APN if and only if for each $\{x_1,x_2,x_3,x_4\} \in \cB_n$,
$$
f(x_1)+f(x_2)+f(x_3)+f(x_4) \ne 0.
$$
Namely, the summation of $f$ over each $2$-dimensional flat is non-vanishing. For a function $f: \Ftn \rightarrow \Ftn$, define the set of \emph{vanishing flats} with respect to $f$ as
\begin{align*}
\VB_{n,f}=\{ \{x_1,x_2,& x_3, x_4\} \in \cB_n \mid \\
                           &  f(x_1)+f(x_2)+f(x_3)+f(x_4)=0 \}.
\end{align*}
Note that each vanishing flat $\{x_1,x_2,x_3,x_4\}$ in $\VB_{n,f}$, gives rise to three distinct critical directions $a_1=x_1+x_2$, $a_2=x_1+x_3$ and $a_3=x_1+x_4$, such that $f(x+a_i)+f(x)$ is not $2$-to-$1$ for each $1 \le i \le 3$. Therefore, the set of critical directions $D_f$ follows immediately from $\VB_{n,f}$. Since $\VB_{n,f} \subset \cB_n$, we call the set system $(\Ftn,\VB_{n,f})$ a \emph{partial quadruple system}. Indeed, a partial quadruple system is an instance of the much more general class of configurations named \emph{packings}. For comprehensive surveys about general Steiner systems and packings, we refer to \cite{CM} and \cite{SWY}, respectively.

For an APN function $f$, the associated partial quadruple system is degenerate as the block set $\VB_{n,f}$ is empty. Now, we proceed to justify that the partial quadruple system is a proper measure of the planarity of non-APN functions.

Let $f$ and $g$ be two polynomials from $\F_{2^n}$ to $\F_{2^n}$. They are {\em CCZ-equivalent} \cite{CCZ} if there exists an affine permutation $A$ on $\F_{2}^n \times \F_{2}^n$ such that
\begin{equation}\label{eqn-CCZ}
\left\{\colvec{x}{g(x)} \bigm| x\in\F_{2}^n \right\}=\left\{A \colvec{x}{f(x)} \bigm| x\in\F_{2}^n \right\},
\end{equation}
where by choosing a basis of $\F_{2^n}$ over $\F_2$, the polynomials $f$ and $g$ are regarded as two mappings from $\F_2^n$ to $\F_2^n$. As it is well-known, the CCZ-equivalence preserves the differential uniformity, and therefore the APN property~\cite[Proposition 2]{BCP}. Recall that two partial quadruple systems are \emph{isomorphic} if there exists a bijective mapping between their point sets, which also induces a bijective mapping between their block sets. Next, we will show that the CCZ-equivalence also preserves partial quadruple systems up to isomorphism. Note that we will use the notation $\cong$ to represent two isomorphic partial quadruple systems.

\begin{theorem}
\label{thm-CCZinvariant}
Let $f$ and $g$ be two CCZ-equivalent functions from $\F_{2^n}$ to $\F_{2^n}$, in which
\begin{align*}
&\left\{\colvec{x}{g(x)} \bigm|  x\in\F_{2}^n \right\} \\
=&\left\{ \begin{pmatrix}
A_{11} & A_{12}\\
A_{21} & A_{22}
\end{pmatrix} \colvec{x}{f(x)}+\colvec{u}{v} \bigm|  x\in\F_{2}^n \right\}.
\end{align*}
Then, $\{x_1,x_2,x_3,x_4\} \in \VB_{n,f}$ if and only if $\{y_1,y_2,y_3,y_4\} \in \VB_{n,g}$, where $y_i=A_{11} x_i+A_{12} f(x_i)+u$.
Consequently, $(\Ftn,\VB_{n,f}) \cong (\Ftn,\VB_{n,g})$ and the number of vanishing flats is an invariant under the CCZ-equivalence.
\end{theorem}
\begin{proof}
Suppose $\{x_1,x_2,x_3,x_4\} \in \VB_{n,f}$. Then
\begin{align*}
x_1+x_2+x_3+x_4&=0, \\
f(x_1)+f(x_2)+f(x_3)+f(x_4)&=0.
\end{align*}
Let $T$ be a mapping from $\F_2^n$ to $\F_2^n$, such that $T(x)=A_{11}x+A_{12}f(x)+u$. By the defining relation of the CCZ-equivalence, $T$ induces a permutation on $\Ftn$. Moreover, we let $y_i=T(x_i)$ and $g(y_i)=A_{21} x_i+A_{22} f(x_i)+v$. A direct computation shows that
\begin{align*}
y_1+y_2+y_3+y_4&=0, \\
g(y_1)+g(y_2)+g(y_3)+g(y_4)&=0.
\end{align*}
Therefore, $\{y_1,y_2,y_3,y_4\} \in \VB_{n,g}$ and $T$ induces an injective mapping from $\VB_{n,f}$ to $\VB_{n,g}$. Reversing the above argument, one can show that by applying the inverse of $T$ to a vanishing flat $\{y_1,y_2,y_3,y_4\} \in \VB_{n,g}$, we obtain a vanishing flat $\{x_1,x_2,x_3,x_4\} \in \VB_{n,f}$. Hence, $T$ induces a bijection between $\VB_{n,f}$ and $\VB_{n,g}$, and therefore $(\Ftn,\VB_{n,f}) \cong (\Ftn,\VB_{n,g})$.
\end{proof}

For a polynomial $f$ over $\Ftn$ and $\sig \in \Gal(\Ftn/\F_2)$ (the Galois group of $\Ftn$ over $\F_2$), we define  $\sig f$ to be a function satisfying  $\sig f(a)=\sig(f(a))$, for each $a \in \Ftn$. It is easy to see that both applying a Galois automorphism of the field extension $\Ftn/\F_2$ and adding an affine function keep the partial quadruple systems invariant. In addition, if $f$ is a permutation, then the partial quadruple systems associated with $f$ and $f^{-1}$ are isomorphic.

\begin{remark}\label{rem-transformation}
Let $f$ be a polynomial over $\Ftn$. Then we have the following$:$
\begin{itemize}
\item[$(1)$] For each $\sig \in \Gal(\Ftn/\F_2)$, we have $\VB_{n,f}=\VB_{n,\sig f}$.
\item[$(2)$] For each affine polynomial $A$ over $\Ftn$, we have $\VB_{n,f}=\VB_{n,f+A}$.
\item[$(3)$] If $f$ is a permutation, then $(\Ftn,\VB_{n,f}) \cong (\Ftn,\VB_{n,f^{-1}})$.
\end{itemize}
\end{remark}


A particularly interesting quantity associated with $(\Ftn,\VB_{n,f})$ is the size of the block set $\VB_{n,f}$, or equivalently, the number of vanishing flats of $f$. Note that the number of vanishing flats is $0$ if and only if $f$ is APN. In this sense, the number of vanishing flats measures the distance between $f$ and the set of APN functions. In fact, the size of $\VB_{n,f}$ follows from the differential spectrum of $f$.

\begin{theorem}
\label{thm-diffspecVB}
For a function $f: \Ftn \rightarrow \Ftn$, we have $\displaystyle \left|\VB_{n,f}\right|=\frac{1}{3}\sum_{a \in \Ftn^*, b \in \Ftn} \binom{\de_f(a,b)/2}{2}$.
\end{theorem}
\begin{proof}
By definition, each pair $(a,b)$ with $\de_f(a,b) \le 2$ is irrelevant to vanishing flats. Thus, we only need to consider the pairs $(a,b) \in \Ftn^* \times \Ftn$ with $\de_f(a,b) \ge 4$. For each such pair, there exist distinct $x_1, x_2 \in \Ftn$, such that $f(x_1+a)+f(x_1)=f(x_2+a)+f(x_2)=b$ and therefore $\{x_1,x_1+a,x_2,x_2+a\}$ forms a vanishing flat. In total, the number of vanishing flats derived from these $(a,b)$ pairs is equal to $\sum_{a \in \Ftn^*, b \in \Ftn} \binom{\de_f(a,b)/2}{2}$. On the other hand, we claim that each vanishing flat can be derived exactly three times by considering the $(a,b)$ pairs. Indeed, for distinct $a_1$, $a_2$, $a_3$ with $a_1+a_2+a_3=0$, a vanishing flat $\{x,x+a_1,x+a_2,x+a_3\}$ satisfying
\begin{align*}
f(x+a_1)+f(x)&=f(x+a_2)+f(x+a_3)=b_1, \\
f(x+a_2)+f(x)&=f(x+a_1)+f(x+a_3)=b_2, \\
f(x+a_3)+f(x)&=f(x+a_1)+f(x+a_2)=b_3,
\end{align*}
occurs exactly three times when we consider the pairs $(a_i,b_i)$, $1 \le i \le 3$. Consequently, $|\VB_{n,f}|=\frac{1}{3}\sum_{a \in \Ftn^*, b \in \Ftn} \binom{\de_f(a,b)/2}{2}$.
\end{proof}

Observe that both the differential uniformity and the number of vanishing flats carry partial information of the differential spectrum. In general, they offer different viewpoints to assess the planarity of functions. On the other hand, when a function has two-valued differential spectrum, these three notions are equivalent. We recall that a function $f$ over $\Ftn$ with two-valued differential spectrum has the property that $\{ \de_f(a,b) \mid a \in \Ftn^*, b \in \Ftn \}=\{0,2^s\}$ for some positive integer $1 \le s \le n$ \cite[Lemma 3]{BCC}. In the following result,  for the sake of simplicity, we only concentrate on the nonzero entries in the differential spectrum.

\begin{proposition}\label{prop-equivalence}
Let $f$ be a polynomial over $\Ftn$ with two-valued differential spectrum $\{0,2^s\}$, where $1 \le s \le n$. Then the following are equivalent$:$
\begin{itemize}
\item[$(1)$] $\de_f=2^s$.
\item[$(2)$] $|\VB_{n,f}|=\frac{2^{n-2}(2^n-1)(2^{s-1}-1)}{3}$.
\item[$(3)$] $\ell_{f,0}=(2^n-2^{n-s})(2^n-1)$ and $\ell_{f,2^s}=2^{n-s}(2^n-1)$.
\end{itemize}
\end{proposition}
\begin{proof}
The proof follows easily from the definitions and Theorem~\ref{thm-diffspecVB}.
\end{proof}

For a function $f$ over $\Ftn$, the number of its vanishing flats tells the number of blocks in the partial quadruple system $(\Ftn,\VB_{n,f})$. As one may expect, the differential spectrum of $f$ supplies more detailed structural information of $(\Ftn,\VB_{n,f})$. Indeed, for $x \in \Ftn$ and $a \in \Ftn^*$ such that $f(x+a)+f(x)=b$, the value of $\de_f(a,b)$ renders the number of vanishing flats containing the pair $x$ and $x+a$.

\begin{theorem}\label{thm-diffspecPQS}
Let $f$ be a function over $\Ftn$.
\begin{itemize}
\item[$(1)$] For $x \in \Ftn$ and distinct $a_1, a_2 \in \Ftn^*$, if the three elements $x$, $x+a_1$ and $x+a_2$ belong to a vanishing flat, then $\de_f(a) \ge 4$ for each $a \in \{ a_1, a_2, a_1+a_2 \}$.
\item[$(2)$] For $x \in \Ftn$ and $a \in \Ftn^*$, such that $f(x+a)+f(x)=b$, the two elements $x$ and $x+a$ belong to $\frac{\de_f(a,b)}{2}-1$ vanishing flats.
\end{itemize}
\end{theorem}
\begin{proof}
The proof of (1) is easy and hence omitted. For (2), note that there are $\frac{\de_f(a,b)}{2}$ pairs of $x_i$ and $x_i+a$, where $1 \le i \le \frac{\de_f(a,b)}{2}$ and $x_1=x$, such that $f(x_i+a)+f(x_i)=b$. Consequently, there are exactly $\frac{\de_f(a,b)}{2}-1$ vanishing flats containing $x$ and $x+a$.
\end{proof}

\begin{remark}
We note that the converse of Theorem~\ref{thm-diffspecPQS}(1) is not true. For instance, consider the inverse function $f(x)=x^{-1}$ over $\F_{2^n}$ with $n$ being even. According to Proposition~\ref{prop-mono} below, $\de_f(a) \ge 4$ for each $a \in \Ftn^*$. On the other hand, by Theorem~\ref{thm-PQSinv}, for distinct $a_1,a_2 \in \Ftn$, the three elements $0$, $a_1$ and $a_2$ belong to the same vanishing flat if and only if $\frac{a_1}{a_2}$ is a third root of unity in $\F_{2^n}$.
\end{remark}

\section{Vanishing Flats and Partial Quadruple Systems of Monomials and Dembowski-Ostrom Polynomials}\label{sec3}

In this section, we investigate the vanishing flats and the partial quadruple systems associated with monomials and Dembowski-Ostrom polynomials.

\subsection{Vanishing Flats and Partial Quadruple Systems of Monomials}

We first consider the vanishing flats and the partial quadruple systems derived from monomials. For the sake of convenience, we write $\de_d:=\de_{x^d}$, $\de_d(a,b):=\de_{x^d}(a,b)$, $\de_d(a):=\de_{x^d}(a)$, $D_d:=D_{x^d}$ and $\VB_{n,d}:=\VB_{n,x^d}$. A polynomial $f$ being monomial ensures that the first-order derivative $f(x+a)+f(x)$ behaves uniformly at each $a \in \Ftn^*$.

\begin{proposition}
\label{prop-mono}
Let $x^d$ be a non-APN monomial over $\Ftn$. Then the multiset $[\de_d(a,b) \mid b \in \Ftn]$ is the same for each $a \in \Ftn^*$. Consequently, $D_d=\Ftn \sm \{0\}$ and
\begin{equation}\label{eqn-VBf}
|\VB_{n,d}| \ge \begin{cases}
              \frac{2^n+1}{3} & \mbox{if $n$ odd} \\
              \frac{2^n-1}{3} & \mbox{if $n$ even}.
            \end{cases}
\end{equation}
\end{proposition}
\begin{proof}
For distinct $a_1,a_2 \in \Ftn^*$ and $b \in \Ftn$, $x$ is a solution to $(x+a_1)^d+x^d=b$ if and only if $\frac{a_2}{a_1}x$ is a solution to $(x+a_2)^d+x^d=\left(\frac{a_2}{a_1}\right)^db$. Hence, $\de_d(a_1,b)=\de_d\left(a_2,\left(\frac{a_2}{a_1}\right)^db\right)$ and consequently, the multiset $[\de_d(a,b) \mid b \in \Ftn]$, and therefore $\de_d(a)$ is the same for each $a \in \Ftn^*$. Since $x^d$ is not APN, we have $\de_d(a) \ge 4$ for each $a \in \Ftn^*$ and thus, $D_d=\Ftn \sm \{0\}$. Namely, $x^d$ has $2^n-1$ critical directions. Note that along each critical direction, there exists at least one vanishing flat. Also, each vanishing flat gives three distinct critical directions. To cover all $2^n-1$ critical directions, one needs at least $\frac{2^n-1}{3}$ vanishing flats. Hence, we have $|\VB_{n,d}| \ge \frac{2^n-1}{3}$, which leads to~\eqref{eqn-VBf}.
\end{proof}

\begin{remark}\label{rem-upperbound}
In addition to \eqref{eqn-VBf}, since $\VB_{n,f} \subset \cB_n$, we have a natural upper bound
$$
|\VB_{n,f}| \le |\cB_n|=\frac{2^{n-2}(2^{n-1}-1)(2^n-1)}{3}.
$$
The upper bound is tight as $\VB_{n,1}=\cB_n$ and $|\VB_{n,1}|=\frac{2^{n-2}(2^{n-1}-1)(2^n-1)}{3}$.
\end{remark}

Equation \eqref{eqn-VBf} gives a general lower bound on the number of vanishing flats for monomials. In fact, when $n$ is even, the lower bound is tight.

\begin{theorem}
\label{thm-PQSinv}
Let $n$ be even and $3 \nmid d$. Let $\al$ be a primitive element of $\Ftn$ and $\ze=\al^{\frac{2^n-1}{3}}$. Then $\VB_{n,d} \supseteq \left\{ \left\{0,\al^i,\al^i\ze,\al^i\ze^2\right\} \mid 0 \le i \le \frac{2^n-4}{3} \right\}$. In particular, for the inverse function $x\mapsto x^{-1}:=x^{2^n-2}$, we have $\VB_{n,-1}=\left\{ \left\{0,\al^i,\al^i\ze,\al^i\ze^2\right\} \mid 0 \le i \le \frac{2^n-4}{3} \right\}$.
\end{theorem}
\begin{proof}
If $3 \nmid d$, then by definition, $\VB_{n,d} \supseteq \left\{ \left\{0,\al^i,\al^i\ze,\al^i\ze^2\right\} \mid 0 \le i \le \frac{2^n-4}{3} \right\}$. Now, we proceed to show that each vanishing flat in $\VB_{n,-1}$ contains $0$, by reduction to absurdity. Assume that $\{x_1,x_2,x_3,x_4\}$ is a vanishing flat with $x_1x_2x_3x_4 \ne 0$. Combining $x_1+x_2=x_3+x_4$ and $x_1^{-1}+x_2^{-1}=x_3^{-1}+x_4^{-1}$, we have $x_1+x_2=x_3+x_4=a$ and $x_1x_2=x_3x_4=b$ for some nonzero $a$ and $b$. Hence, $x_i$, $1 \le i \le 4$, are four distinct solutions  to the quadratic equation $x^2+ax+b=0$, which is impossible.

Without loss of generality, we assume $x_4=0$ and $x_3=\be$ for some nonzero $\be$. Hence, we have $x_1+x_2=\be$ and $x_1^{-1}+x_2^{-1}=\be^{-1}$. It is easy to see that $\{x_1,x_2\}=\{\be\ze,\be\ze^2\}$. Therefore, two elements $0$ and $\be$ determine the whole vanishing flat $\{x_1,x_2,x_3,x_4\}=\left\{0,\be,\be\ze,\be\ze^2\right\}$. In conclusion, $\VB_{n,-1}=\left\{ \left\{0,\al^i,\al^i\ze,\al^i\ze^2\right\} \mid 0 \le i \le \frac{2^n-4}{3} \right\}$.
\end{proof}

The partial quadruple system can also be determined for the non-APN Gold function.  We shall use $(a,b)$ to denote the greatest common divisor of $a$ and $b$.

\begin{theorem}\label{thm-PQSGold}
Let $f(x)=x^{2^t+1}$ be a function over $\Ftn$ with $(n,t)=s>1$. For $a \in \Fts \sm \{0,1\}$ and $x \in \Ftn^*$, we define a $2$-dimensional vector space $V_{a,x}=\{0,x,ax,(1+a)x\}$ and
\begin{align*}
U_{a,x}=\{ & \left\{c,x+c,ax+c,(1+a)x+c\right\} \mid \mbox{$c$ ranges} \\
& \mbox{over all coset representatives of $V_{a,x}$ in $\Ftn$} \}.
\end{align*}
Then $\displaystyle \VB_{n,f}=\bigcup_{\substack{a \in \Fts \sm \{0,1\} \\ x \in \Ftn^*}} U_{a,x}$ and $\displaystyle |\VB_{n,f}|=\frac{2^{n-2}(2^{s-1}-1)(2^n-1)}{3}$.
\end{theorem}
\begin{proof}
Note that if $\{x_1,x_2,x_3,x_4\}$ is a vanishing flat of $f$, then so is $\{x_1+c,x_2+c,x_3+c,x_4+c\}$ for each $c \in \Ftn$. Hence, it suffices to analyze the $2$-dimensional vanishing subspaces, which have the form $\{0,x_1,x_2,x_3\}$, where the nonzero $x_1$, $x_2$, $x_3$ are distinct, $x_1+x_2+x_3=0$ and $f(0)+f(x_1)+f(x_2)+f(x_3)=x_1^{2^t+1}+x_2^{2^t+1}+x_3^{2^t+1}=0$.  Note that $x_1^{2^t+1}+x_2^{2^t+1}+x_3^{2^t+1}=x_1^{2^t+1}+x_2^{2^t+1}+(x_1+x_2)^{2^t+1}=0$. Simplifying the last equation, we have $(\frac{x_1}{x_2})^{2^t}=\frac{x_1}{x_2}$. Thus, $\frac{x_1}{x_2} \in \Fts$. Since $x_1 \ne 0$ and $x_1 \ne x_2$, we have $\frac{x_1}{x_2} \in \Fts \sm \{0,1\}$. Indeed, we have shown for each $1 \le j<\ell\le 3$, $\frac{x_j}{x_\ell} \in \Fts \sm \{0,1\}$. Note that each $2$-dimensional vanishing subspace $\{0,x_1,x_2,x_3\}$ can be expressed in exactly six different ways: $V_{\frac{x_2}{x_1},x_1}$, $V_{\frac{x_3}{x_1},x_1}$, $V_{\frac{x_1}{x_2},x_2}$, $V_{\frac{x_3}{x_2},x_2}$, $V_{\frac{x_1}{x_3},x_3}$ and $V_{\frac{x_2}{x_3},x_3}$. Consequently, we have $\VB_{n,f}=\bigcup_{\substack{a \in \Fts \sm \{0,1\} \\ x \in \Ftn^*}} U_{a,x}$. Moreover, $|\VB_{n,f}|=\frac{2^{n-2}(2^s-2)(2^n-1)}{6}=\frac{2^{n-2}(2^{s-1}-1)(2^n-1)}{3}$.
\end{proof}

So far, we have determined the partial quadruple systems associated with the inverse function (Theorem~\ref{thm-PQSinv}) and the Gold  functions (Theorem~\ref{thm-PQSGold}). It is worthy to note that there are a few power functions whose differential spectra are known. Therefore, employing Theorem~\ref{thm-diffspecVB}, the number of vanishing flats follows immediately. We summarize these results in Table~\ref{tab-VF}, where $s=(n,t)$ and $K$ is the Kloosterman sum with the following explicit expression \cite[Remark 3]{BCC11}:
$$
K=1+\frac{(-1)^{n-1}}{2^{n-1}}\sum_{i=0}^{\lf \frac{n}{2} \rf} (-1)^i \binom{n}{2i}7^i.
$$
For two positive integers $a$ and $b$, we define
$$
\De(a,b)=\begin{cases}
  1 & \mbox{if $a \mid b$,} \\
  0 & \mbox{if $a \nmid b$.}
\end{cases}
$$
By Proposition~\ref{prop-mono}, for every monomial, each entry among the differential spectrum is divisible by $2^n-1$. For simplicity, we shall use $w_i=\ell_i/(2^n-1)$ in Table~\ref{tab-VF}. The column ``Ref'' indicates the references where the differential spectra are calculated.

\begin{table*}\renewcommand{\arraystretch}{1.5}

\begin{center}
\caption{Power functions $x^d$ over $\Ftn$ with known differential spectra and number of vanishing flats}

\begin{tabular}{|c|c|c|c|c|}
\hline
$n$ & $d$ & $|\VB_{n,d}|$ & Differential Spectrum &  Reference \\ \hline
\multirow{2}{*}{$n \ge 2$} & $2^t+1$ & \multirow{2}{*}{$\dfrac{2^{n-2}(2^{s-1}-1)(2^n-1)}{3}$} & $w_0=2^n-2^{n-s}$ & \multirow{2}{*}{\cite[Section 5.2]{BCC}} \\
& $1 \le t \le n/2$ & & $w_{2^s}=2^{n-s}$ & \\ \hline
$n \ne 3t$ & $2^{2t}-2^t+1$ & \multirow{2}{*}{$\dfrac{2^{n-2}(2^{s-1}-1)(2^n-1)}{3}$} & $w_0=2^n-2^{n-s}$ & \multirow{2}{*}{\cite[Theorem 2]{BCC}} \\
$n/s$ odd & $2 \le t \le n/2$ & & $w_{2^s}=2^{n-s}$ & \\ \hline
\multirow{3}{*}{$n$ even}& \multirow{3}{*}{$2^n-2$} & \multirow{3}{*}{$\dfrac{2^n-1}{3}$} & $w_0=2^{n-1}+1$ & \multirow{3}{*}{\cite[Proposition 6]{N}} \\
 & & & $w_{2}=2^{n-1}-2$ & \\
 & & & $w_{4}=1$ & \\ \hline
\multirow{3}{*}{$n=4t$}& \multirow{3}{*}{$2^{2t}+2^t+1$} & \multirow{3}{*}{$\dfrac{\left(2^{n-3}-2^{3t-3}\right)(2^n-1)}{3}$} & $w_0=5\cdot2^{n-3}-2^{3t-3}$ & \multirow{2}{*}{\cite[Example 4]{BCC}} \\
 & & & $w_{2}=2^{n-2}+2^{3t-2}$ & \\
 & & & $w_{4}=2^{n-3}-2^{3t-3}$ & \cite[Theorem 1]{XY}\\ \hline
\multirow{4}{*}{$n \ge 6$}& \multirow{4}{*}{$7$} & & $w_0=2^{n-1}+2w_6+w_4$ & \multirow{4}{*}{\cite[Theorem 5]{BCC11}} \\
& & $\left(\frac{2^{n-2}+1-3w_4}{6}+(-1)^n\frac{K}{8}\right)$ & $w_{2}=2^{n-1}-3w_6-2w_4$ & \\
& & $\cdot (2^n-1)$ & $w_{4}=\De(2,n)$ & \\
& & & $w_{6}=\frac{2^{n-2}+1-5w_4}{6}+(-1)^n\frac{K}{8}$ & \\ \hline
\multirow{5}{*}{$n \ge 6$}& $2^{n-2}-1$ &  & $w_0=2^{n-1}+2w_6+3w_8$ & \\
& or & $\Big(\frac{2^{n-1}-3-(-1)^n5}{12}+(-1)^n\frac{K}{8}$ & $w_{2}=2^{n-1}-3w_6-4w_8$ & \cite[Corollary 5]{BCC11} \\
& $2^{\frac{n-1}{2}}-1$ & $+w_8\Big)\cdot(2^n-1)$ & $w_{6}=\frac{2^{n-1}-3-(-1)^n5}{12}$ & \cite[Theorem 5]{BP} \\
& $n$ odd & & $+(-1)^n\frac{K}{8}-w_8$ & \\
&  & & $w_{8}=\De(3,n)$ & \\ \hline
& \multirow{4}{*}{$2^{\frac{n}{2}}-1$} &  & $w_0=2^{n-1}+2^{\frac{n}{2}-1}-2+w_4$ & \multirow{4}{*}{\cite[Theorem 7]{BCC11}} \\
$n \ge 6$ & & $\left(\left(2^{\frac{n}{2}-1}-1\right)\Big(2^{\frac{n}{2}-2}-1\Big)\right.$ & $w_{2}=2^{n-1}-2^{\frac{n}{2}-1}+1-2w_4$ & \\
$n$ even & & $+w_4\Big)\cdot\frac{2^n-1}{3}$ & $w_{4}=1-\De(4,n)$ & \\
& & & $w_{2^{\frac{n}{2}}-2}=1$ & \\ \hline
& \multirow{3}{*}{$2^{\frac{n}{2}+1}-1$} &  & $w_0=2^{n-1}+2^{\frac{n}{2}-1}-1$ & \multirow{3}{*}{\cite[Theorem 8]{BCC11}} \\
$n \ge 6$  & & $\dfrac{2^{\frac{n}{2}-2}\left(2^{\frac{n}{2}-1}-1\right)(2^n-1)}{3}$ & $w_{2}=2^{n-1}-2^{\frac{n}{2}-1}$ & \\
$n$ even & & & $w_{2^{\frac{n}{2}}}=1$ & \\
\hline
& \multirow{3}{*}{$2^{\frac{n+3}{2}}-1$} & & $w_0=2^{n-1}+2w_6+2\De(3,n)$ & \multirow{3}{*}{\cite[Theorems 1,5]{BP}} \\
$n \ge 7$  & & $\left(\frac{2^{n-2}+1}{6}-\frac{K}{8}\right)(2^n-1)$ & $w_{2}=2^{n-1}-3w_6-3\De(3,n)$ & \\
$n$ odd & & & $w_6=\frac{2^{n-2}+1}{6}-\frac{K}{8}$ & \\ \hline
& & & $w_0=89 \cdot 2^{n-7}+7\cdot 2^{t-7}(4-K)$ & \multirow{5}{*}{\cite[Theorem 1.4]{XYY}} \\
$n=2t$ & $2^{t+1}+2^{\frac{t+1}{2}}+1$ & & $w_{2}=5 \cdot 2^{n-5}-5 \cdot 2^{t-5}(4-K)$ & \\
$t \ge 5$ & or & $\dfrac{2^{n-2}(2^n-1)}{3}$ & $w_{4}=7 \cdot 2^{n-6}+9 \cdot 2^{t-6}(4-K)$ & \\
$t$ odd & $2^{t+1}+3$ &  & $w_{6}=2^{n-5}-2^{t-5}(4-K)$ & \\
& & & $w_{8}=2^{n-7}-2^{t-7}(4-K)$ & \\ \hline
\end{tabular}
\label{tab-VF}
\end{center}
\end{table*}

\begin{remark}
The following are some observations based on Table~\textup{\ref{tab-VF}}, regarding the relationship between the differential uniformity, the number of vanishing flats and the differential spectrum.
\begin{itemize}
\item[$(1)$] By Proposition~\textup{\ref{prop-equivalence}}, for the first two classes, the differential uniformity, the number of vanishing flats and the differential spectrum carry the same information.
\item[$(2)$] Comparing with the differential uniformity, there are some cases where the number of vanishing flats is a finer criterion to evaluate the planarity. For example, the third and fourth classes have the same differential uniformity but distinct number of vanishing flats. The same happens to the fifth and sixth classes when $3 \nmid n$.
\item[$(3)$] By Theorem~\textup{\ref{thm-diffspecVB}}, the number of vanishing flats follows from the differential spectrum. The converse is not true as the number of vanishing flats does not contain the information of $w_0$ and $w_2$. For example, whether $3\, |\, n$ or not, the differential spectra of the ninth class differ at $w_0$ and $w_2$. However, the number of vanishing flats remains the same.
\end{itemize}
\end{remark}

In Table~\ref{tab-VFexample}, we list the number of vanishing flats of all monomials over $\F_{2^n}$, where $2 \le n \le 8$. We use the superscript $\bigstar$ to mark the monomials whose number of vanishing flats is not obtained from Table~\ref{tab-VF} or Remark~\ref{rem-upperbound}. Note that by Remark~\ref{rem-transformation}, the number of vanishing flats of $x^d$ is the same as $x^{2^id}$, for each $0 \le i \le n-1$, as well as, $x^{2^id^{-1}}$ for each $0 \le i \le n-1$, if $(d,2^n-1)=1$. Among these monomials which necessarily have the same number of vanishing flats, we only list one of them.

\begin{table*}\renewcommand{\arraystretch}{1.5}

\begin{center}
\caption{The number of vanishing flats of $x^d$ over $\Ftn$, for $2 \le n \le 8$}
\begin{tabular}{|c|c|}
\hline
$n$ & $\left(d,|\VB_{n,d}|\right)$ \\ \hline
$2$ & $(1,1)$ \\ \hline
$3$ & $(1,14)$, $(3,0)$ \\ \hline
$4$ & $(1,140)$, $(3,0)$, $(5,20)$, $(7,5)$ \\ \hline
$5$ & $(1,1240)$, $(3,0)$, $(5,0)$, $(15,0)$ \\ \hline
\multirow{2}{*}{$6$} & $(1,10416)$, $(3,0)$, $(5,336)$, $(7,84)$, $(9,1008)$, $(11,336)^\bigstar$,  \\
                     & $(15,126)$, $(21,2520)^\bigstar$, $(27,1260)^\bigstar$, $(31,21)$  \\ \hline
$7$ & $(1,85344)$, $(3,0)$, $(5,0)$, $(7,889)$, $(9,0)$, $(11,0)$, $(19,889)^\bigstar$, $(21,889)$, $(23,0)$, $(63,0)$ \\ \hline
\multirow{4}{*}{$8$} & $(1,690880)$, $(3,0)$, $(5,5440)$, $(7,3655)$, $(9,0)$, $(11,5185)^\bigstar$, $(13,5185)^\bigstar$, $(15,1785)$,\\
                     & $(17,38080)$, $(19,4420)^\bigstar$, $(21,2040)$, $(23,4930)^\bigstar$, $(25,4420)^\bigstar$, $(27,15810)^\bigstar$, \\
                     & $(31,2380)$, $(39,0)$, $(43,27625)^\bigstar$, $(45,1785)^\bigstar$, $(51,66300)^\bigstar$, $(53,7480)^\bigstar$, $(55,5440)^\bigstar$,  \\
                     &  $(63,3570)$, $(85,174760)^\bigstar$, $(87,24480)^\bigstar$, $(95,2380)^\bigstar$, $(111,1020)^\bigstar$, $(119,41905)^\bigstar$, $(127,85)$
                 \\ \hline
\end{tabular}
\label{tab-VFexample}
\end{center}
\end{table*}

Finally, we mention a connection between the vanishing flats of monomials and the binary cyclic codes defined by the same monomials. For a more detailed account on such cyclic codes, we refer the reader to~\cite{CCZ,CTZ}.

\begin{proposition}\label{prop-vfcdoe}
Let $\al$ be a primitive element of $\Ftn$. Let $C_d$ be a binary cyclic code of length $2^n-1$ having two zeroes $\al$ and $\al^d$, in which each nonzero element of $\Ftn$ indexes a coordinate of the codewords in $C_d$. Each vanishing flat of $x^d$ must be one of the following$:$
\begin{itemize}
\item[$(1)$] For each weight three codeword, the zero element plus the three nonzero elements indexing three nonzero coordinates form a vanishing flat.
\item[$(2)$] For each weight four codeword, the four nonzero elements indexing four nonzero coordinates form a vanishing flat.
\end{itemize}
Consequently, the number of weight three codewords of $C_d$ equals the number of vanishing flats in $\VB_{n,d}$ containing $0$, and the number of weight four codewords of $C_d$ equals the  number of vanishing flats in $\VB_{n,d}$ not containing $0$.
\end{proposition}
\begin{proof}
First, we note that the binary cyclic code $C_d$ has a parity-check matrix
$$
\begin{pmatrix}
1 & \al & \al^2 & \cdots & \al^{2^n-2} \\
1 & \al^d & \al^{2d} & \cdots & \al^{(2^n-2)d}
\end{pmatrix}.
$$
The rest follows easily from the definition of vanishing flats.
\end{proof}

\begin{remark}
Proposition~\ref{prop-vfcdoe} can be further extended. Let $f$ be a polynomial over $\Ftn$. Using the same notation as in Proposition~\ref{prop-vfcdoe}, there is a one-to-one correspondence between the vanishing flats of $f$ and the indices of the weight four codewords in the binary linear code having a parity-check matrix
$$
\begin{pmatrix}
0 & 1 & \al & \al^2 & \cdots & \al^{2^n-2} \\
f(0) & f(1) & f(\al) & f(\al^{2}) & \cdots & f(\al^{2^n-2})
\end{pmatrix}.
$$
\end{remark}

\subsection{Vanishing Flats and Partial Quadruple Systems of Dembowski-Ostrom Polynomials}

A polynomial $f$ over $\Ftn$ is of \emph{Dembowski-Ostrom} (DO) type if
$$
f=\sum_{0 \le i<j<n} c_{ij}x^{2^i+2^j},
$$
where $c_{ij} \in \Ftn$. For DO polynomials, we have the following result on the number of vanishing flats.

\begin{theorem}
\label{thm-cosetstrucDO}
Let $\displaystyle f(x)=\sum_{0 \le i<j <n}c_{i,j}x^{2^i+2^j}$ be a non-APN DO polynomial over $\Ftn$, $n \ge 2$. The following hold$:$
\begin{itemize}
\item[$(1)$] If $\{x_1,x_2,x_3,x_4\} \in \VB_{n,f}$, then $\{\{x_1+a,x_2+a,x_3+a,x_4+a\} \mid a \in \Ftn\} \subset \VB_{n,f}$ for each $a \in \Ftn$. Consequently, $2^{n-2} \bigm| |\VB_{n,f}|$ and $|\VB_{n,f}| \ge 2^{n-2}$.
\item[$(2)$] For each $a \in \Ftn$, the subset $\{a,x_1+a,x_2+a,x_1+x_2+a\} \in \VB_{n,f}$ if and only if
$$
\sum_{0 \le i<j <n}c_{i,j}\left(x_1^{2^i}x_2^{2^j} + x_1^{2^j}x_2^{2^i}\right) = 0.
$$
\end{itemize}
\end{theorem}
\begin{proof}
(1) Let $\displaystyle f(x)=\sum_{0 \le i<j <n} c_{ij}x^{2^i+2^j}$. Since $\{x_1,x_2,x_3,x_4\} \in \VB_{n,f}$, we have $x_1+x_2+x_3+x_4=0$ and $f(x_1)+f(x_2)+f(x_3)+f(x_4)=0$. For $a \in \Ftn$ and $1 \le \ell \le 4$, we have
\begin{align*}
f(x_\ell+a)=&\sum_{0 \le i<j <n} c_{ij}(x_\ell+a)^{2^i+2^j}\\
=&\sum_{0 \le i<j <n} c_{ij} \left(x_\ell^{2^i+2^j}+a^{2^i}x_\ell^{2^j}+a^{2^j}x_\ell^{2^i}+a^{2^i+2^j}\right)\\
=& f(x_{\ell})+f(a)+\sum_{0 \le i<j <n} c_{ij}\left(a^{2^i}x_\ell^{2^j}+a^{2^j}x_\ell^{2^i}\right)
\end{align*}
Note that $\displaystyle \sum_{\ell=1}^4 c_{ij}a^{2^i}x_\ell^{2^j}=c_{ij}a^{2^i}\left(\sum_{\ell=1}^4 x_\ell\right)^{2^j}=0$ and $\displaystyle \sum_{\ell=1}^4 c_{ij}a^{2^j} x_\ell^{2^i}=c_{ij}a^{2^j}\left(\sum_{\ell=1}^4 x_\ell\right)^{2^i}=0$. Therefore, we have $f(x_1+a)+f(x_2+a)+f(x_3+a)+f(x_4+a)=0$. Hence, for each $a \in \Ftn$, $\{x_1+a,x_2+a,x_3+a,x_4+a\} \in \VB_{n,f}$. Namely, $\VB_{n,f}$ contains all $2^{n-2}$ cosets of the flat $\{x_1,x_2,x_3,x_4\}$. Indeed, $\VB_{n,f}$ consists of a collection of $2$-dimensional vector spaces and all their cosets, which implies our claim.

(2) By the first part, $\{a,x_1+a,x_2+a,x_1+x_2+a\} \in \VB_{n,f}$ if and only if $\{0,x_1,x_2,x_1+x_2\} \in \VB_{n,f}$. Note that
$$
f(x_1+x_2)=f(x_1)+f(x_2)+\sum_{0 \le i<j <n} c_{ij}\left(x_1^{2^i}x_2^{2^j}+x_1^{2^j}x_2^{2^i}\right).
$$
Hence, $\{0,x_1,x_2,x_1+x_2\} \in \VB_{n,f}$ if and only if $\sum_{0 \le i<j <n} c_{ij}\left(x_1^{2^i}x_2^{2^j}+x_1^{2^j}x_2^{2^i}\right)=0$.
\end{proof}


Let $f(x)=\sum_{0 \le i<j<n}c_{ij}x^{2^i+2^j}$ be a DO polynomial over $\Ftn$. For $a \in \Ftn^*$, we define a linearized polynomial
$$
L_{f,a}(x)=\sum_{0 \le i<j \le n-1}c_{ij} \left(a^{2^i}x^{2^j}+a^{2^j}x^{2^i}\right).
$$
For such a DO polynomial $f$, we next define the multiset
\begin{equation}
\label{eqn-Rf}
R_f=\left[ \rank(L_{f,a}) \mid a \in \Ftn^* \right].
\end{equation}
The following theorem describes how to determine the number of vanishing flats for a DO polynomial.

\begin{theorem}
\label{thm-vanishingnumDO}
Let $f$ be a DO polynomial. Then $\displaystyle |\VB_{n,f}|=\frac{2^{n-2}}{3}\sum_{h \in R_f}\left(2^{n-h-1}-1\right)$.
\end{theorem}
\begin{proof}
For $a \in \Ftn^*$, we have $f(x+a)+f(x)=L_{f,a}(x)+f(a)$. Hence, to determine the differential spectrum of $f$, it suffices to consider the multiset $[L_{f,a}(x)+f(a) \mid x \in \Ftn]$. Define $h_a=\rank(L_{f,a})$, then the multiset $\left[ L_{f,a}(x)+f(a) \mid x \in \Ftn \right]$ contains $2^{h_a}$ elements each with multiplicity $2^{n-h_a}$. Thus, $\sum_{b \in \Ftn} \binom{\de_f(a,b)/2}{2}=2^{h_a}\binom{2^{n-h_a}/2}{2}=2^{n-2}\left(2^{n-h_a-1}-1\right)$. Hence, by Theorem~\ref{thm-diffspecVB}, we have $|\VB_{n,f}|=\frac{1}{3}\sum_{{\substack a \in \Ftn^* \\ b \in \Ftn}} \binom{\de_f(a,b)/2}{2}=\frac{2^{n-2}}{3}\sum_{h \in R_f}\left(2^{n-h-1}-1\right)$.
\end{proof}

\section{Partitioning vector spaces into disjoint equidimensional affine subspaces via vanishing flats}
\label{sec4}

Given a vector space over a finite field, a long-standing problem is to find a partition of that vector space into a collection of equidimensional vector subspaces, see for instance~\cite{B,HKK}. As an example, for a vector space $V=\Fq^{2n}$ over $\Fq$, a collection of $q^n+1$ subspaces of $V$ with dimension $n$, which contains each nonzero element of $V$ exactly once, is the well-known configuration named \emph{spread}, which has been intensively studied in finite geometry (see~\cite{LP}, for example). A perhaps less known partition was proposed by Baum and Neuwirth~\cite{BN75}, which aims to partition a vector space over a finite field into a collection of disjoint affine subspaces with equal dimension. In this section, we present an application of vanishing flats which generates this kind of partition in an elegant way. We call an $r$-dimensional affine subspace an $r$-flat. The following is a definition of~\cite{BN75}.

\begin{definition}
Let $n \ge 2$ and $1 \le d \le n-1$. A {\em cover} of $\F_{2}^n$ with dimension $d$ is a partition of $\F_{2}^n$ into $2^{n-d}$ disjoint affine subspaces each having dimension $d$. A $d$-dimensional linear subspace of $\F_2^n$ and all its cosets forms a {\em trivial cover}.
\end{definition}

Let $\{A_i\}_{i=1}^{2^{n-d}}$ and $\{B_i\}_{i=1}^{2^{n-d}}$ be two covers of $\F_2^n$ with dimension $d$. They are \emph{equivalent} if there exists an affine permutation $T$, such that $T(A_i)=B_i$, where the order of the affine subspaces in $\{B_i\}_{i=1}^{2^{n-d}}$ can be adjusted if necessary. Clearly, two trivial covers with the same dimension are equivalent. A major question is, are there any covers inequivalent to the trivial ones? In~\cite{BN75}, Baum and Neuwirth proposed to classify all covers up to equivalence, but that seems to be elusive, so far. On the other hand, they made considerable progress along this direction and in particular, they provided constructions of nontrivial covers. To describe their results, we need a few more concepts.

For an affine subspace $A_i$, we use $A_i^*$ to denote the associated linear subspace. A cover $\{A_i\}_{i=1}^{2^{n-d}}$ of $\F_2^n$ is called \emph{nonparallel} if $A_i^* \ne A_j^*$ for every distinct $i$ and $j$. A cover $\{A_i\}_{i=1}^{2^{n-d}}$ of $\F_2^n$ is called \emph{totally skew} if $A_i^* \cap A_j^*=\{0\}$ for every distinct $i$ and $j$. In \cite{BN75}, the following nontrivial covers were obtained.

\begin{result}\label{res-nonpara}
For each $n \ge 3$, there exists a nonparallel cover of $\F_2^n$ with dimension $d$, where $1 \le d \le n-2$. Note that each cover of $\F_2^n$ with dimension $(n-1)$ necessarily consists of two parallel $(n-1)$-flats. Hence, there exists a nonparallel cover of $\F_2^n$ with dimension $d$ if and only if $1 \le d \le n-2$.
\end{result}

\begin{result}\label{res-skew}
For each $n \ge 3$, there exists a totally skew cover of $\F_2^n$ with dimension $d$, where $1 \le d \le \frac{n-1}{2}$.
\end{result}

%
%


In this section, we employ DO monomials to supply new constructions of totally skew covers. Again, we use the finite field $\Ftn$ to represent the vector space $\F_2^n$. Given a polynomial $f$ over $\Ftn$ and $A=\{x_1,x_2,x_3,x_4\} \in \VB_{n,f}$, define $f(A)=\{ f(x_1), f(x_2), f(x_3), f(x_4)\}$. Moreover, for a permutation polynomial $f$ and a subset $\cA \subset \VB_{n,f}$, we define
\begin{align*}
f(\cA)=\{ \{ f(x_1), f(x_2),  f(x_3), & f(x_4) \} \mid \\
                                      &\{ x_1, x_2, x_3, x_4\} \in \cA\}.
\end{align*}
The next theorem outlines the strategy of generating covers by using permutation polynomials and their vanishing flats.

\begin{theorem}
\label{thm-covergene}
Let $f$ be a permutation polynomial over $\Ftn$. Suppose $\cA \subset \VB_{n,f}$ is a cover of $\Ftn$ with dimension $2$. Then, so is $f(\cA)$.
\end{theorem}
\begin{proof}
The property of vanishing flats implies that $f(\cA)$ is a collection of $2$-dimensional affine subspaces. The permutation property ensures that $f(\cA)$ is a cover.
\end{proof}

The above theorem indicates that given a cover with dimension $2$ contained in the vanishing flats of a permutation polynomial, a new cover can be derived by applying the permutation. Naturally, the most interesting permutations are the DO permutations, whose vanishing flats are formed by a union of trivial covers with dimension $2$:

\begin{theorem}
\label{thm-coverDO}
Let $f$ be a DO permutation over $\Ftn$, which is not APN. Then $\VB_{n,f}$ is a disjoint union of trivial covers $\cA_i$, $\displaystyle 1 \le i \le \frac{1}{3}\sum_{h \in R_f}\left(2^{n-h-1}-1\right)$, with dimension $2$,
where
$$
\cA_i=\{ \{ c,x_i+c,y_i+c,x_i+y_i+c \} \mid c \in \Ftn \}
$$
and $R_f$ is defined in \eqref{eqn-Rf}. The set $f(\cA_i)$ is a cover with dimension $2$ for each $i$. In addition, $f(\cA_i)$ is totally skew if and only if the following hold$:$
\begin{itemize}
\item[$(1)$] $\de_f(x_i)=\de_f(y_i)=\de_f(x_i+y_i)=4$,
\item[$(2)$] $E_f(x_i)$, $E_f(y_i)$ and $E_f(x_i+y_i)$ are pairwise disjoint.
\end{itemize}
\end{theorem}
\begin{proof}
The first part follows from Theorems~\ref{thm-cosetstrucDO}(1), \ref{thm-vanishingnumDO} and \ref{thm-covergene}. For distinct $c, d \in \Ftn$, write
\begin{align*}
A_i(c)&=\{c, x_i+c, y_i+c, x_i+y_i+c\} \in \cA_i,  \\
A_i(d)&=\{d, x_i+d, y_i+d, x_i+y_i+d\} \in \cA_i.
\end{align*}
Note that $f(\cA_i)$ is totally skew if and only if for every distinct $c,d \in \Ftn$, $f(A_i(c))^* \cap f(A_i(d))^*=\{0\}$. Namely,
$\{0, f(x_i+c)+f(c), f(y_i+c)+f(c), f(x_i+y_i+c)+f(c)\}$ and $\{0, f(x_i+d)+f(d), f(y_i+d)+f(d), f(x_i+y_i+d)+f(d)\}$ intersect trivially at $0$, for every distinct $c,d \in \Ftn$.  This is equivalent to the two conditions on the top of the next page.

\begin{figure*}[t]
{\flushleft (a) For each $c \in \Ftn$ and $d \in \Ftn \sm \{ c, x_i+c, y_i+c, x_i+y_i+c \}$,}
\begin{equation*}
\begin{aligned}
f(x_i+c)+f(c) &\ne f(x_i+d)+f(d), \\
f(y_i+c)+f(c) &\ne f(y_i+d)+f(d), \\
f(x_i+y_i+c)+f(c) &\ne f(x_i+y_i+d)+f(d).
\end{aligned}
\end{equation*}

{\flushleft (b) For each $c \in \Ftn$,}
\begin{equation*}
\begin{aligned}
&f(x_i+c)+f(c) \notin \{ f(y_i+d)+f(d), f(x_i+y_i+d)+f(d) \mid d \in \Ftn \sm \{ c, x_i+c, y_i+c, x_i+y_i+c \} \}, \\
&f(y_i+c)+f(c) \notin \{ f(x_i+d)+f(d), f(x_i+y_i+d)+f(d) \mid d \in \Ftn \sm \{ c, x_i+c, y_i+c, x_i+y_i+c \} \}, \\
&f(x_i+y_i+c)+f(c) \notin \{ f(x_i+d)+f(d), f(y_i+d)+f(d) \mid d \in \Ftn \sm \{ c, x_i+c, y_i+c, x_i+y_i+c \} \}.
\end{aligned}
\end{equation*}
\hrulefill
\end{figure*}

Noting that $f(c)+f(x_i+c)+f(y_i+c)+f(x_i+y_i+c)=0$, we observe that Condition (a) is equivalent to $\de_f(x_i)=\de_f(y_i)=\de_f(x_i+y_i)=4$ and Condition (b) is equivalent to $E_f(x_i)$, $E_f(y_i)$ and $E_f(x_i+y_i)$ being pairwise disjoint.
\end{proof}

By restricting to DO monomials, we expect more specific information about the associated covers.  We note that each DO monomial is CCZ-equivalent to a Gold function. In view of Theorem~\ref{thm-CCZinvariant}, we only need to consider covers derived from Gold functions. For $f(x)=x^{2^t+1}$ over $\Ftn$, $f$ is a permutation if and only if $(2^t+1,2^n-1)=1$ and is not APN if and only if $(n,t)>1$. We note that $(2^t+1,2^n-1)=1$ is equivalent to $\frac{n}{(n,t)}$ being odd. As a preparation, we have the following crucial lemma which follows from \cite[Theorem 3]{CH}.

\begin{lemma}
\label{lem-numsol}
Let $x \in \Ftn$. Suppose $s=(t,n)$ and $z \in \Fts^*$. Then the equation
$$
x^{2^t}+x=z
$$
has no solution in $\Ftn$ if and only if $\frac{n}{(n,t)}$ is odd. Moreover, suppose that $\frac{n}{(n,t)}$ is odd, then for $x,y \in \Ftn$ and $z_1, z_2 \in \Fts$, the equation
\begin{equation}
\label{eqn-numsol}
x^{2^t}+x+z_1=y^{2^t}+y+z_2
\end{equation}
holds if and only if $z_1=z_2$ and $x,y$ belong to the same additive coset of $\Fts$ in $\Ftn$.
\end{lemma}
\begin{proof}
 According to \cite[Theorem 3]{CH}, the equation $x^{2^t}+x=z$ has no solution in $\Ftn$ if and only if $\sum_{i=0}^{\frac{n}{(n,t)}-1} z^{2^{ti}} \ne 0$. Since $z \in \Fts^*$, $s=(t,n)$ and $\frac{n}{(n,t)}$ odd, we have $\sum_{i=0}^{\frac{n}{(n,t)}-1} z^{2^{ti}}=z \ne 0$ and the first part follows. For the second part, the sufficiency is clear. Note that $x^{2^t}+x+z_1=y^{2^t}+y+z_2$ is equivalent to $(x+y)^{2^t}+(x+y)=z_1+z_2$. Since $\frac{n}{(n,t)}$ is odd, then~\eqref{eqn-numsol} holds only if $z_1=z_2$, which forces $x,y$ being in the same additive coset of $\Fts$ in $\Ftn$.
\end{proof}

Now we proceed to consider DO monomials.

\begin{corollary}\label{cor-coverGold}
Let $f(x)=x^{2^t+1}$ be over $\Ftn$, where $1 \le t \le n-1$, $(2^t+1,2^n-1)=1$ and $s=(n,t)>1$. Then we have
\begin{itemize}
\item[$(1)$] $\VB_{n,2^t+1}$ is a disjoint union of trivial covers $\left\{\cA_i \mid 1 \le i \le \frac{(2^{s-1}-1)(2^n-1)}{3} \right\}$ with dimension $2$, where
    $$
      \cA_i=\left\{ \{ c,x_i+c,y_i+c,x_i+y_i+c \} \mid c \in \Ftn \right\},
    $$
    for $x_i, y_i \in \Ftn^*$ and $\frac{x_i}{y_i} \in \Fts \sm \{0,1\}$. For each $i$, $f(\cA_i)$ is a cover with dimension $2$.
\item[$(2)$] For $c \in \Ftn$, write
    $$
      A_i(c)=\{c, x_i+c, y_i+c, x_i+y_i+c\} \in \cA_i,
    $$
    where $x_i, y_i \in \Ftn^*$ and $\frac{x_i}{y_i} \in \Fts \sm \{0,1\}$. Write $x_i=\al z_i$ and $y_i=\al w_i$ for some $\al \in \Ftn^*$ and distinct $z_i, w_i \in \Fts^*$. Let $b_{\al,j}+\al\Fts$, $1 \le j \le 2^{n-s}$, be $2^{n-s}$ additive cosets of $\al\Fts$ in $\Ftn$.
    Then $f(A_i(c))$ is parallel to $f(A_i(d))$ if and only if $c,d \in b_{\al,j}+\al\Fts$, for some $j$. Otherwise, $f(A_i(c))^* \cap f(A_i(d))^*=\{0\}$. In particular, $f(\cA_i)$ can be partitioned into $2^{n-s}$ subsets $\bigcup_{c \in b_{\al,j}+\al\Fts}f(A_i(c))$, $1 \le j \le 2^{n-s}$, where each subset contains $2^{s-2}$ parallel $2$-flats and for every pair of $2$-flats from distinct subsets, their associated linear subspace intersect trivially. Consequently, for each $i$, $f(\cA_i)$ is a totally skew cover with dimension $2$ if $s=2$ and is a nontrivial cover with dimension $2$ having parallel $2$-flats, if $s>2$.
\end{itemize}
\end{corollary}
\begin{proof}
We only need to show Part (2), as Part (1) follows from Theorems~\ref{thm-PQSGold}, \ref{thm-cosetstrucDO}, and \ref{thm-coverDO}. Note that $x_i=\al z_i$, $y_i=\al w_i$, where $z_i,w_i \in \Fts^*$ are distinct. Since $z_i^{2^t}=z_i$ and $w_i^{2^t}=w_i$, we can compute $f(A_i(c))^*$ and $f(A_i(d))^*$ in the middle of the next page.

\begin{figure*}[t]
\begin{align*}
f(A_i(c))^*&=\{0, f(x_i+c)+f(c), f(y_i+c)+f(c), f(x_i+y_i+c)+f(c)\} \\
           &=\left\{0, x_ic^{2^t}+x_i^{2^t}c+x_i^{2^t+1}, y_ic^{2^t}+y_i^{2^t}c+y_i^{2^t+1}, (x_i+y_i)c^{2^t}+(x_i+y_i)^{2^t}c+(x_i+y_i)^{2^t+1} \right \} \\
           &=\left\{0, \al z_i^{2^t} c^{2^t}+\al^{2^t}z_ic+\al^{2^t+1}z_i^2, \al w_i^{2^t} c^{2^t}+\al^{2^t}w_ic+\al^{2^t+1}w_i^2\right., \\
           & \quad \quad \quad \quad \quad \quad \quad \quad \quad\left. \al (z_i+w_i)^{2^t} c^{2^t}+\al^{2^t}(z_i+w_i)c+\al^{2^t+1}(z_i+w_i)^2 \right\} \\
           &=\Bigg\{0, \al^{2^t+1}\left(\left(\frac{z_ic}{\al}\right)^{2^t}+\frac{z_ic}{\al}+z_i^2\right), \al^{2^t+1}\left(\left(\frac{w_ic}{\al}\right)^{2^t}+\frac{w_ic}{\al}+w_i^2\right), \\
         & \quad \quad \quad \quad \quad \quad \quad \quad \quad \al^{2^t+1}\left(\left(\frac{(z_i+w_i)c}{\al}\right)^{2^t}+\frac{(z_i+w_i)c}{\al}+(z_i+w_i)^2\right) \Bigg\}
\end{align*}
and similarly,
\begin{align*}
f(A_i(d))^*&=\Bigg\{0, \al^{2^t+1}\left(\left(\frac{z_i d}{\al}\right)^{2^t}+\frac{z_i d}{\al}+z_i^2\right), \al^{2^t+1}\left(\left(\frac{w_i d}{\al}\right)^{2^t}+\frac{w_i d}{\al}+w_i^2\right) \\
         & \quad \quad \quad \quad \quad \quad \quad \quad \quad \al^{2^t+1}\left(\left(\frac{(z_i+w_i)d }{\al}\right)^{2^t}+\frac{(z_i+w_i) d}{\al}+(z_i+w_i)^2\right) \Bigg\}.
\end{align*}
\hrulefill
\end{figure*}

Since $z_i, w_i \in \Fts^*$ are distinct, by Lemma~\ref{lem-numsol}, $|f(A_i(c))^* \cap f(A_i(d))^*| \ge 2$ if and only if $\left(\frac{\be c}{\al}\right)^{2^t}+\frac{\be c}{\al}+\be^2=\left(\frac{\be d}{\al}\right)^{2^t}+\frac{\be d}{\al}+\be^2$ for some $\be \in \{ z_i, w_i, z_i+w_i \}$. Each of these three cases is equivalent to $\left(\frac{c+d}{\al}\right)^{2^t}=\frac{c+d}{\al}$. Hence, $f(A_i(c))^*=f(A_i(d))^*$ if and only if $c,d \in b_{\al,j}+\al\Fts$, for some $j$. Otherwise, $f(A_i(c))^* \cap f(A_i(d))^*=\{0\}$. Consequently, $f(\cA_i)$ can be partitioned into $2^{n-s}$ subsets $\bigcup_{1 \le j \le 2^{n-s}} (\bigcup_{c \in b_j+\al\Fts} f(A_i(c)))$, where for every pair of $2$-flats from distinct subsets, their associated linear spaces intersect trivially. Note that $f(A_i(c))=f(A_i(d))$ if and only if $c+d \in \{ 0,x_i,y_i,x_i+y_i \}$. Thus, each subset $\bigcup_{c \in b_j+\al\Fts} f(A_i(c))$ contains $\frac{2^s}{4}=2^{s-2}$ parallel $2$-flats.
\end{proof}

Consequently, we have the following necessary and sufficient condition determining which DO monomial leads to totally skew covers.

\begin{corollary}\label{cor-skewcoverGold}
Let $f(x) = x^{2^t+1}$ be over $\Ftn$, where $(2^t+1,2^n-1)=1$ and  $s=(n,t)>1$. For each trivial cover
$$
\cA_i=\{ \{ c,x_i+c,y_i+c,x_i+y_i+c \} \mid c \in \Ftn \},
$$
where $x_i, y_i \in \Ftn^*$ and $\frac{x_i}{y_i} \in \Fts \sm \{0,1\}$, the cover $f(\cA_i)$ is totally skew if and only if $n \equiv 2 \pmod 4$ and $s=2$. More precisely, let $\ga$ be a primitive element of $\Ftn$. Then each trivial cover $\cA_j=\{ \{ c+\ga^j\F_4 \} \mid c \in \Ftn\}$, where $0 \le j \le \frac{2^n-4}{3}$, gives a totally skew cover $f(\cA_j)$.
\end{corollary}
\begin{proof}
We only need to prove the necessary and sufficient condition, as the rest follows immediately. The sufficiency follows from Corollary~\ref{cor-coverGold}(2). Conversely, if $f(\cA_i)$ is totally skew, by Corollary~\ref{cor-coverGold}(2), $s=(n,t)=2$. Assume $n \equiv 0 \pmod 4$. Then we have $t \equiv 2 \pmod 4$. In this case, $(2^t+1,2^n-1)=5$, which gives a contradiction. Hence, we must have $n \equiv 2 \pmod 4$.
\end{proof}

We observe that the concept of vanishing flats can be extended to higher dimensions. Let $f$ be over $\Ftn$, an $s$-flat in $\F_{2^n}$ is called a \emph{vanishing $s$-flat} of $f$, if its image under $f$ is again an $s$-flat. Clearly, a vanishing flat of $f$ is a vanishing $2$-flat. The following theorem generalizes Corollary~\ref{cor-skewcoverGold} and produces totally skew covers with higher dimensions.
For a subset $A \subset \Ftn$ and a function $f$ over $\Ftn$, we write $f(A)=\{f(a) \mid a \in A\}$.

\begin{theorem}\label{thm-coverhigh}
Let $f(x)=x^{2^t+1}$ be over $\Ftn$, with $(2^t+1,2^n-1)=1$ and $s=(n,t)>1$. For every $ \al \in \Ftn^*$, the set $\{ c+\al\Fts \mid c \in \Ftn \}$ forms a trivial cover of $\Ftn$ with dimension $s$, which is a vanishing $s$-flat of $f$. Moreover, $\{ f(c+\al\Fts) \mid c \in \Ftn \}$ is a totally skew cover of $\Ftn$ with dimension $s$.
\end{theorem}
\begin{proof}
Note that

\begin{align*}
 &f(c+\al\Fts) \\
=&\left\{ c^{2^t+1}+c^{2^t}\al z+c\al^{2^t}z^{2^t}+\al^{2^t+1}z^{2^t+1} \mid z \in \Fts \right\} \\
=&c^{2^t+1}+S_c,
\end{align*}
where
$$S_c=\left\{(c^{2^t}\al+c\al^{2^t})z+\al^{2^t+1}z^2 \mid z \in \Fts\right\}.$$
Clearly, $S_c$ is a linear subspace. Since $f$ is a permutation, the dimension of $S_c$ is $s$. Consequently, $f(c+\al\Fts)$ is an $s$-flat and hence, $c+\al\Fts$ is a vanishing $s$-flat, for each $c \in \Ftn$. To prove that $\{ f(c+\al\Fts) \mid c \in \Ftn \}$ is a totally skew cover, it suffices to show that $S_c=S_d$ if and only if $c+d \in \al\Fts$ and $S_c \cap S_d= \{0\}$, otherwise.

 Assume $S_c \cap S_d$ contains a nonzero element. Then there exist $z_1, z_2 \in \Fts^*$, such that
\begin{equation}\label{eqn-S}
(c^{2^t}\alpha+c\alpha^{2^t})z_1+\alpha^{2^t+1}z_1^2=(d^{2^t}\alpha+d\alpha^{2^t})z_2+\alpha^{2^t+1}z_2^2.
\end{equation}
Rewrite \eqref{eqn-S} as
$$
(c^{2^t}z_1+d^{2^t}z_2)\alpha+(cz_1+dz_2)\alpha^{2^t}+(z_1+z_2)^2\alpha^{2^t+1}=0.
$$
Noting that $z_1^{2^t}=z_1$ and $z_2^{2^t}=z_2$, we have
$$
(cz_1+dz_2)^{2^t}\alpha+(cz_1+dz_2)\alpha^{2^t}+(z_1+z_2)^2\alpha^{2^t+1}=0,
$$
and therefore,
$$
(\frac{cz_1+dz_2}{\alpha})^{2^t}+\frac{cz_1+dz_2}{\alpha}=(z_1+z_2)^2.
$$
Since $\frac{cz_1+dz_2}{\alpha} \in \Ftn$ and $z_1, z_2 \in \Fts^*$, by Lemma~\ref{lem-numsol}, the above equation holds only if $z_1=z_2$. By Equation~\eqref{eqn-S}, $z_1=z_2$ implies $(\frac{c+d}{\alpha})^{2^t}=\frac{c+d}{\alpha}$, and therefore, $c+d\in\al\Fts$. Conversely, if $c+d\in\al\Fts$, a direct computation shows $S_c=S_d$.
\end{proof}

\begin{remark}
Assume the conditions of Theorem~\textup{\ref{thm-coverhigh}} hold. Let $A$ be a $2$-dimensional subspace of $\al\Fts$. Suppose $\cA=\{A_i \mid 1 \le i \le 2^{n-2} \}$ is a trivial cover of $\Ftn$ formed by $A$ and all its cosets. By Corollary~\textup{\ref{cor-coverGold}(2)}, as a cover of $\Ftn$, $f(\cA)$ can be partitioned into $2^{n-s}$ subsets, where each subset contains $2^{s-2}$ parallel $2$-flats. By taking a union of all the parallel $2$-flats in each subset, we exactly recover the totally skew cover of $\Ftn$ with dimension $s$ derived in Theorem~\textup{\ref{thm-coverhigh}}.
\end{remark}

\section{Concluding remarks}\label{sec5}

In this paper, we investigated the planarity of a function from the viewpoint of vanishing flats. We showed that in some cases, the number of vanishing flats is a criterion that is finer than the differential uniformity and coarser than the differential spectrum. Moreover, the partial quadruple system formed by all the vanishing flats contains detailed information about $f$. For monomials and DO polynomials, we studied their number of vanishing flats and partial quadruple systems in more detail. As an application, we extended the concept of vanishing flats and proposed a new construction of totally skew covers using DO monomials.

We believe that the vanishing flats and the partial quadruple systems deserve further investigation. We propose next some further questions:
\begin{itemize}
\item[$(1)$] The lower bound in \eqref{eqn-VBf} is tight when $n$ is even, which is achieved by the inverse function. Are there any other monomials achieving this lower bound? Another natural question is to consider is whether~\eqref{eqn-VBf} is tight when $n$ is odd. If so, characterize the monomials achieving this lower bound.
\item[$(2)$] By Theorem~\ref{thm-cosetstrucDO}(1), the number of vanishing flats of each non-APN DO polynomial over $\Ftn$ is lower bounded by $2^{n-2}$. According to Theorem~\ref{thm-PQSGold}, DO monomials cannot achieve this lower bound. On the other hand, is there any DO polynomial having exactly $2^{n-2}$ vanishing flats?
\item[$(3)$] So far, we only know the partial quadruple systems associated with the inverse and Gold functions. It is interesting to determine the partial quadruple systems associated with other polynomials.
\item[$(4)$] Instead of the partial quadruple systems, an easier problem is to determine the number of vanishing flats associated with other polynomials. For instance, the open cases in Table~\ref{tab-VFexample} may serve as a starting point.
\item[$(5)$] In Corollary~\ref{cor-skewcoverGold}, the $\frac{2^n-1}{3}$ trivial covers correspond to $\frac{2^n-1}{3}$ totally skew covers. It is interesting to consider the equivalence problem of these totally skew covers. Moreover, we ask whether the totally skew covers derived from DO monomials are equivalent to those in \cite[Theorem 2]{BN75}.
\item[$(6)$] In Corollary~\ref{cor-coverGold}, the totally skew covers were derived from DO monomials whenever $n \equiv 2 \pmod 4$. On the other hand, there are some known DO permutations other than monomials \cite{BCHO}. Thus, we ask if any totally skew covers can be derived from these DO permutations, especially when $n \not\equiv 2 \pmod 4$.
\end{itemize}

\section*{Acknowledgement}

Shuxing Li is supported by the Alexander von Humboldt Foundation. Wilfried Meidl is supported by the FWF Project P 30966. This paper was started while the second, fifth and  sixth-named authors visited the fourth-named author at Otto von Guericke University Magdeburg in the spring of 2019. These authors would like to thank the host institution for the excellent working conditions.

\begin{IEEEbiographynophoto}{Shuxing Li}
received his Ph.D. degree in Mathematics from Zhejiang University, China, in 2016. From November 2014 to
July 2016, He was a research assistant at Department of Mathematics, The Hong Kong University of Science and Technology, Hong Kong. From September 2016 to September 2017, He was a postdoctoral fellow at Department of Mathematics, Simon Fraser University, Canada. He was an Alexander-von-Humboldt Postdoctoral Fellow from October 2017 to September 2019, at Faculty of Mathematics, Otto von Guericke University Magdeburg, Germany. Currently, he is a postdoctoral fellow funded by the Pacific Institute for the Mathematical Sciences, at Department of Mathematics, Simon Fraser University, Canada. His research interests include algebraic and combinatorial design theory, algebraic coding theory, and finite geometry. In 2018, he received the Kirkman Medal from the Institute of Combinatorics and its Applications.
\end{IEEEbiographynophoto}

\begin{IEEEbiographynophoto}
{Wilfried Meidl} received the Ph.D. degree from Klagenfurt University, Austria, in 1998. From 2000 to 2002 he was with the Institute of Discrete
Mathematics, OEAW, Vienna, Austria. From 2002--2004 he was with Temasek Labs, National University of Singapore, and from 2005--2014 he was with
Sabanc\i University, Istanbul, Turkey. He is now with RICAM, OEAW, Linz, Austria. His research interests include sequences, permutation polynomials,
finite fields and their applications, Boolean functions, bent functions.
\end{IEEEbiographynophoto}

\begin{IEEEbiographynophoto}{Alexandr Polujan}
received the Diploma in mathematics and system analysis from the Belarusian State University, Minsk, Belarus, in 2015.
Since 2016, he has been pursuing the Ph.D. degree under the supervision of Alexander Pott at the Otto von Guericke University, Magdeburg, Germany. His research interests include mainly perfect non-linear functions.
\end{IEEEbiographynophoto}

\begin{IEEEbiographynophoto}{Alexander Pott}
received his Ph.D. degree in Mathematics in 1988 from the Justus-Liebig-University Giessen. He held visiting positions in Dayton (Ohio) and Duisburg (Germany). He was professor at the University of Augsburg (Germany), and presently he is full professor at the Otto-von-Guericke-University Magdeburg, Germany. His research interests include Finite Geometry (difference sets), Algebraic Coding Theory (sequences, Boolean functions) and Finite Fields.
\end{IEEEbiographynophoto}

\begin{IEEEbiographynophoto}{Constanza Riera}
received her Master of Science degree in 2002 from Autonomous University of Madrid, Spain. She completed her Ph.D. in Mathematics at Complutense University of Madrid, Spain in 2006. Currently, she is an Associate Professor at Western Norway University of Applied Sciences in Bergen, Norway. Her research interests are Cryptology, Sequence Design, Coding Theory and Quantum Communications.
\end{IEEEbiographynophoto}

\begin{IEEEbiographynophoto}{Pantelimon St\u anic\u a}
received his Master of Science in Mathematics degree in 1992 from  University of Bucharest, Romania. He completed his Ph.D. in Mathematics at  State University of New York at Buffalo in 1998. Currently, he is a Professor at the Naval Postgraduate School, in Monterey, California. His research interests are in Cryptology, Coding Theory, Sequence Design, Number Theory and Discrete  Mathematics.
\end{IEEEbiographynophoto}

\end{document}